\documentclass[12pt]{article}

\usepackage[left=1in, top=1in, right=1in, bottom=1in]{geometry}
\usepackage[semibold]{libertine}
\usepackage{amsmath}
\usepackage{amssymb}
\usepackage{amsthm}
\usepackage{paralist}
\usepackage{algorithmicx}
\usepackage{mathrsfs}
\usepackage{dsfont}
\usepackage[noend]{algpseudocode}
\usepackage[linesnumbered,lined,boxed,commentsnumbered]{algorithm2e}
\usepackage{color}
\usepackage[square,numbers]{natbib}
\usepackage[pdfpagelabels,pdfpagemode=None]{hyperref}
\usepackage{aliascnt}
\usepackage{tcolorbox}
\usepackage{siunitx}
\usepackage{cleveref}

\usepackage{mathtools}
\usepackage{bm}
\usepackage{todonotes}
\usepackage{thmtools} 
\usepackage{thm-restate}

\newtheorem{theorem}{Theorem}%
\newtheorem{lemma}{Lemma}%
\newtheorem{fact}{Fact}

\newtheorem{definition}{Definition}

\newtheorem*{remark}{Remark}
\newtheorem{property}{Property}

\usepackage[utf8]{inputenc}

\newcommand{\Xcomment}[1]{{}}

\newcommand{\noaccents}[1]{#1}

\newcommand{\newagentvar}[3][\noaccents]{%
\expandafter\newcommand\expandafter{\csname #2\endcsname}{#1{#3}}%
\expandafter\newcommand\expandafter{\csname #2s\endcsname}{#1{\boldsymbol{#3}}}%
\expandafter\newcommand\expandafter{\csname #2smi\endcsname}[1][i]{#1{\boldsymbol{#3}}_{-##1}}%
\expandafter\newcommand\expandafter{\csname #2i\endcsname}[1][i]{#1{#3}_{##1}}%
\expandafter\newcommand\expandafter{\csname #2k\endcsname}[1][i]{#1{#3}_{##1}}%
\expandafter\newcommand\expandafter{\csname #2ith\endcsname}[1][i]{#1{#3}_{(##1)}}%
}

\DeclareMathOperator*{\argmax}{arg\,max}

\makeatletter

\newcommand{\taolin}[1]{\textcolor{blue}{[Tao says\@ifnotempty{#1}{: #1}]}}

\makeatother

\newcommand{\opt}{\mathcal{M}^*}

\DeclareMathOperator{\PM}{PM}
\DeclareMathOperator{\sw}{SW}

\newagentvar{sample}{x}
\newagentvar{val}{v}
\newagentvar{bid}{b}
\newagentvar{strat}{s}
\newagentvar{tstrat}{t}

\newagentvar[\tilde]{sampleval}{\val}

\title{On Clearing Prices in Matching Markets: A Simple Characterization without Duality}

\author{Xiaoming Li \\ Peking University \\  \href{mailto:lxm@pku.edu.cn}{lxm@pku.edu.cn} 
   \and Tao Lin\footnote{Authors are listed alphabetically.} \\ Peking University \\ \href{mailto:lin_tao@pku.edu.cn}{lin\_tao@pku.edu.cn} }

\usepackage{natbib}
\usepackage{graphicx}

\begin{document}

\maketitle

\begin{abstract}
Duality of linear programming is a standard approach to the classical weighted maximum matching problem. From an economic perspective, the dual variables can be regarded as prices of products and payoffs of buyers in a two-sided matching market. Traditional duality-based algorithms, e.g., Hungarian, essentially aims at finding a set of prices that clears the market. Under such market-clearing prices, a maximum matching is formed when buyers buy their most preferred products respectively. We study the property of market-clearing prices without the use of duality, showing that: (1) the space of market-clearing prices is convex and closed under element-wise maximum and minimum operations; (2) any market-clearing prices induce all maximum matchings. 
\end{abstract}

\section{Introduction}
The \emph{weighted bipartite matching problem}, or \emph{assignment problem}, is a classical problem in graph theory: how can we find a matching with the maximum total weight (maximum matching for short) in a bipartite graph? There is a direct connection between a maximum matching and a socially optimal allocation of products to buyers in a \emph{matching market}. In this market, there are $n$ buyers and $n$ product. Each product can be allocated to at most one buyer, and each buyer wants at most one product. Let $V=(v_{ij})\in\mathbb{R}_+^{n\times n}$ be a matrix of valuations where buyer $i$ thinks product $j$ is worth $v_{ij}$ dollars. Clearly, valuations constitute a complete weighted bipartite graph $G$ with $n$ left nodes representing buyers, $n$ right nodes representing products, and $n^2$ edges where $v_{ij}$ is the weight of edge $(i, j)$. A valid allocation of products to buyers correspond to a \emph{matching} of $G$. The weight or social welfare of a matching is the sum of $v_{ij}$ over the buyer-product pair $(i, j)$ in it. So, how can we find a matching that maximizes the social welfare? 

\paragraph{Duality and matching market}
A standard approach to the maximum matching problem is the duality of linear programming.
Let $x_{ij}\in\{0, 1\}$ denote whether $i$ gets $j$. First write down our goal in a linear program:  
\begin{align*}
    \max_{x_{ij}}~~ & \sum_{i=1}^n\sum_{j=1}^n v_{ij}x_{ij}, \\
    \text{s.t.}~~ & \sum_{j=1}^n x_{ij} = 1 ~~ (i=1, ..., n), \\
    & \sum_{i=1}^n x_{ij} = 1 ~~ (j=1, ..., n), \\
    & x_{ij}\in\{0, 1\}, \forall i, j\\
    & \text{ (or replace by $x_{ij}\ge0$, which does not affect the solution).}
\end{align*}
Then derive its dual program: 
\begin{align*}
    \min_{u_i, p_j}~~ & \sum_{i=1}^n u_i + \sum_{j=1}^n p_j, \\
    \text{s.t.}~~ & u_i + p_j \ge v_{ij}, ~~\forall i,j.
\end{align*}

Many algorithms for the maximum matching problem, including the famous Hungarian algorithm proposed in the seminal work \cite{kuhn1955hungarian} of Kuhn in 1955, are based on duality. Those algorithms update the dual variables $u_i$ and $p_j$ dynamically to obtain a solution to the dual program which can then be converted to a maximum matching in the original linear program.

Furthermore, there is a natural economic interpretation of the duality approach. Dual variables on the product side, $p_j$, can be regarded as \emph{prices}. If buyer $i$ buys product $j$, she pays $p_j$ and obtains a \emph{payoff} of $v_{ij}-p_j$. Dual variables on the buyer side, $u_i$, then become the maximum payoff buyer $i$ can obtain by selecting her most favorable product. This connection between dual variables and prices and payoffs has been discovered by Shapley and Shubik in 1971 \cite{shapley1971assignment} and some following works, e.g., \cite{roth1992two}.

\paragraph{Market-clearing prices} Although the duality approach and its economic interpretation are standard, a beginner who is not familiar with linear programming still cannot understand those complicated duality-based algorithms like Hungarian easily. Can we solve the maximum matching problem with simpler techniques? Fortunately, an intuitive version of Hungarian algorithm has been discussed in the textbook \textit{Networks, Crowds, and Markets} by Easley and Kleinberg \cite{easley2010networks}, solely in the language of prices and payoffs, without mentioning dual variables at all. Once we focus on prices and payoffs, the picture becomes vivid: The algorithm in \cite{easley2010networks} increases the price of a product when too many buyers want it, so the payoffs of that product to buyers decrease accordingly. Adding an edge to a matching corresponds to a buyer buying a product. 

Moreover, when the above algorithm ends, the resulting prices obtain a nice property called \emph{market-clearing} (Definition \ref{def:market-clearing_price}). For such a price vector $\bm{p}=(p_1, ..., p_n)$, the market clears in the sense that each buyer is able to buy her most favorable product (the product that maximizes her payoff $v_{ij}-p_j$), without conflicting with other buyers. Thus, a market-clearing price vector induces a socially optimal allocation of products to buyers, exactly maximizing the weight of the matching (Property \ref{property:included}). In order words, finding a maximum matching in a weighted bipartite graph boils down to finding a clearing price vector in a matching market. 

The market-clearing price vector may not be unique. Actually, as argued in \cite{shapley1971assignment}, the set of all market-clearing price vectors is exactly the $p_j$-part of the solution set to the dual linear program mentioned above, and thus has some properties like convexity. What other properties do market-clearing prices have? Can we show those properties \emph{without}  the help of duality? 

\paragraph{Our results} We ask two questions concerning market-clearing prices: 
\begin{enumerate}
    \item If there are multiple market-clearing price vectors, what properties does the set of market-clearing price vectors have? 
    \item If there are multiple matchings with the maximum weight in a bipartite graph, can each and every maximum matching be induced by some market-clearing price vector?
\end{enumerate}

We give the following answers:
\begin{enumerate}
    \item The set of market-clearing prices is closed under the operation of convex combination, element-wise maximum, and element-wise minimum (Theorem \ref{thm:space_of_prices} in Section \ref{sec:space_of_prices}).
    \item Yes. \emph{Any} market-clearing price vector induces \emph{all} maximum matchings in a bipartite graph (Theorem \ref{thm:all_maximum_matching} in Section \ref{sec:any_to_all}).
\end{enumerate}

Our proofs are elementary and do not involve the duality of linear programming at all. Instead, our arguments rely on a key lemma (Lemma \ref{lem:two_prices}) which states that all market-clearing price vectors induce a same set of maximum matchings. Then our first theorem follows almost immediately. The proof of our second theorem requires some additional knowledge of graph theory. We believe that Lemma \ref{lem:two_prices} is of independent interest. We also believe that our arguments can be taught in related courses, or assigned as exercises, even at an undergraduate level, because of their appropriate level of difficulty. 

\paragraph{Related works}
Section 10 in \cite{easley2010networks} is an easy-to-read introduction to the bipartite matching problem. It provides a detailed description of how to update prices to satisfy the market-clearing property. One can refer to \cite{burkard2012assignment} for the duality approach and many algorithms for the bipartite matching problem. \cite{pentico2007assignment} provides a large list of variants of the assignment problem.


\section{Preliminaries}

\paragraph{A matching market}
There are $n$ buyers and $n$ products. Each product $j$ can be allocated (sold) to at most one buyer, and each buyer $i$ wants at most one product. 
Let $V=(v_{ij})\in\mathbb{R}_+^{n\times n}$ be a matrix of valuations where buyer $i$ values product $j$ at $v_{ij}\ge 0$. Valuations constitute a complete weighted bipartite graph $G$ which has $n$ left nodes representing buyers, $n$ right nodes representing products, and $n^2$ edges where the weight of edge $(i, j)$ is $v_{ij}$. Let $p_j$ be the price of product $j$. If a buyer $i$ buys product $j$, she pays $p_j$ and obtains a \emph{payoff} of $v_{ij}-p_j$.


\paragraph{Allocation and social welfare}
Clearly, a valid allocation of products to buyers correspond to a \emph{matching} $M$ of $G$, which is a subset of edges that do not share any common node with each other. Given a matching $M$, we can define the \emph{social welfare} $\sw(M)$ (or the weight of $M$) as the sum of the valuations over all buyers for their received products,
\begin{equation}
    \sw(M) = \sum_{(i, j)\in M} v_{ij}. 
\end{equation}
We will focus mainly on matchings that maximize the social welfare (maximum matchings). Since we have assumed $v_{ij}\ge 0$, it is never worse to choose a \emph{perfect} matching, which means that the size of matching $|M|=n$, i.e., each buyer buys some product (and all products are sold). So, we assume that maximum matchings are always perfect. 


\paragraph{Preferred-product graph}
Given a set of prices $\bm{p}=(p_1, ..., p_n)$, if buyer $i$ wants to maximize her payoff, she will buy a product $j$ that maximizes $v_{ij}-p_j$. Assume that she can choose any product if there are multiple products with the maximum payoff.
We can describe the set of products each buyer is willing to buy by another bipartite graph, named \emph{preferred-product graph}. Given a price vector $\bm{p}=(p_1, ..., p_n)$, define the following sub-graph of $G$
\begin{equation}
G(\bm{p}):=\bigg\{(i, j)\in G \bigg\vert v_{ij}-p_j \ge v_{ik}-p_k, \forall k=1, ..., n\bigg\}
\end{equation}
as the preferred-product graph of $\bm{p}$. 

Preferred-product graphs have a useful property: 
\begin{fact}\label{preferred-product_graph}
A perfect matching (a matching of size $n$) in the preferred-product $G(\bm{p})$, if exists, is a matching in $G$ that maximizes the social welfare. 
\end{fact}
\begin{proof}
Let $M^*$ be a perfect matching in $G(\bm{p})$. 
For any other matching $M'$ in $G$, if $M'$ is not perfect, we complete it to be a perfect matching $M$ by allocating unsold products to unmatched buyers arbitrarily. Note that $\sw(M')\le\sw(W)$ because $v_{ij}\ge 0$. Then: 
\begin{align*}
    \sw(M^*) - \sw(M) = & \sum_{(i, j_i^*)\in M^*} (v_{ij_i^*} - p_{j_i^*} + p_{j_i^*}) - \sum_{(i, j_i)\in M} (v_{ij_i} - p_{j_i} + p_{j_i}) \\
    = & \sum_{i=1}^n p_{j_i^*} - \sum_{i=1}^n p_{j_i} +  \sum_{i=1}^n \bigg( (v_{ij_i^*} - p_{j_i^*}) - (v_{ij_i} - p_{j_i}) \bigg) \\
    \ge & 0 + \sum_{i=1}^n 0, 
\end{align*}
where the inequality follows from the fact that each $(i, j_i^*)$ is in $G(\bm{p})$.
\end{proof}

\begin{definition}[Market-clearing prices]\label{def:market-clearing_price}
A price vector $\bm{p}=(p_1, ..., p_n)$ is \emph{market-clearing}, if its preferred-product graph $G(\bm{p})$ contains a perfect matching. 
\end{definition}

For convenience, we say $\bm{p}$ \emph{induces} a matching $M$ if $M$ is a perfect matching in $G(\bm{p})$. For a price vector $\bm{p}$, let $\PM(\bm{p})$ denote the set of matchings induced by $\bm{p}$. And we use $\opt=\argmax_M\{\sw(M)\}$ to denote the set of all maximum (perfect) matchings in $G$. 
By Definition \ref{def:market-clearing_price}, together with Fact \ref{preferred-product_graph}, we immediately have: 
\begin{property}[Market-clearing prices induce maximum matchings]\label{property:included}
For any market-clearing price vector $\bm{p}$, 
\begin{equation}
     \emptyset\ne \PM(\bm{p})\subseteq \opt.
\end{equation}
\end{property}

\section{Space of Market-Clearing Prices}\label{sec:space_of_prices}
In this section we answer the first question: if there are multiple market-clearing price vectors, what properties does the space of all such price vectors have?
\begin{theorem}[Closure of space]\label{thm:space_of_prices}
Suppose price vectors $\bm{p}$ and $\bm{q}$ are market-clearing. Let $\bm{r}\in\mathbb{R}^n$ be another price vector obtained by transforming $\bm{p}$ and $\bm{q}$ in the following four ways:  
\begin{enumerate}
    \item Diagonal shifting: $\bm{r}=\bm{p} + (t, ..., t)$, for any $t\in\mathbb{R}$. That is, moving $\bm{p}$ in the diagonal direction $\bm{1}=(1, ..., 1)$ forward or backward. 
    \item Convex combination: $\bm{r} = \alpha\bm{p} + (1-\alpha)\bm{q}$, for any $\alpha\in[0, 1]$.
    \item Element-wise maximum: $\bm{r}=(r_1, ..., r_n)$ where $r_i=\max\{p_i, q_i\}$ for each $i=1, ..., n$.
    \item Element-wise minimum: $\bm{r}=(r_1, ..., r_n)$ where $r_i=\min\{p_i, q_i\}$ for each $i=1, ..., n$.
\end{enumerate}
Then $\bm{r}$ is also market-clearing.
\end{theorem}

Now we prove Theorem \ref{thm:space_of_prices}. The first item follows from the definition of preferred-product graph $G(\bm{p})$: 
\begin{itemize}
\item Proof for diagonal shifting. Note that if all prices are increased or decreased by a same amount $t$, the quantitative relation between the payoffs of any two products remains unchanged for each buyer, so $G(\bm{r})=G(\bm{p})$, and $\bm{r}$ is market-clearing as $\bm{p}$ is. 
\end{itemize}

For the other three items, let $S$ be the set of all market-clearing price vectors. The second item in Theorem \ref{thm:space_of_prices} states that $S$ is closed under convex combination. This property is also called the \emph{convexity} of $S$. For readers who are familiar with the duality approach, the convexity of $S$ may seem trivial: given that the space of market-clearing prices is exactly (a projection of) the solution set of the dual program, $S$ must be convex because the programming is linear.

However, the above argument cannot be easily generalized to prove that $S$ is also closed under element-wise maximum and minimum. In order to prove Theorem \ref{thm:space_of_prices} without duality, we first prove the following key lemma which implies that all market-clearing price vectors induce a same set of matchings. 

\begin{lemma}[Same set of induced matchings]\label{lem:two_prices}
For any two market-clearing price vectors $\bm{p}$ and $\bm{q}$,
\[\PM(\bm{p}) = \PM(\bm{q}).\]
\end{lemma}

\begin{proof}
Recall that $\PM(\bm{p}), \PM(\bm{q})$ are the sets of perfect matchings in the preferred-product graphs $G(\bm{p}), G(\bm{q})$. 

Let $M=\{(i_1, j_1), (i_2, j_2), ..., (i_n, j_n)\}$ be a perfect matching in $G(\bm{p})$, we will show that $M$ is also a perfect matching in $G(\bm{q})$. To do this, we only need to prove that the all the edges $(i_1, j_1), ..., (i_n, j_n)$ are in $G(\bm{q})$. 

Prove by contradiction. Suppose that edge $(i_1, j_1)$ is not in $G(\bm{q})$. Since $\bm{q}$ is a market-clearing, there exists a perfect matching $N$ in $G(\bm{q})$. Because $(i_1, j_1)\notin G(\bm{q})$, $i_1$ must be matched to a product other than $j_1$ in $N$. Without loss of generality, assume $i_1$ is matched to $j_2$; then $i_2$ cannot be matched to $j_2$, and we assume $i_2$ is matched to $j_3$; ...; until some $k\le n$, where $i_{k-1}$ is matched to $j_k$, and $i_k$ must be matched back to $j_1$ in $N$. Now, $(i_1, j_2), (i_2, j_3), ..., (i_k, j_1) \in N$. 

Consider $(i_1, j_1)$ and $(i_1, j_2)$. Because $(i_1, j_1)$ is preferred in $G(\bm{p})$, we have $v_{i_1, j_1}-p_{j_1} \ge v_{i_1, j_2}-p_{j_2}$. Because $(i_1, j_2)$ is preferred in $G(\bm{q})$ and $(i_1, j_1)$ is \textbf{NOT} preferred in $G(\bm{q})$, we have: $v_{i_1, j_1}-q_{j_1} < v_{i_1, j_2}-q_{j_2}$. Subtracting the two inequalities gives: 
$$q_{j_2}-q_{j_1} < p_{j_2} - p_{j_1}. $$

Consider other edges: $(i_2, j_2)$ is preferred in $G(\bm{p})$ while $(i_2, j_3)$ is preferred in $G(\bm{q})$, we have: 
$$q_{j_3}-q_{j_2} \le p_{j_3} - p_{j_2}. $$
Similarly, 
$$\vdots$$
$$q_{j_1}-q_{j_k} \le p_{j_1} - p_{j_k}.$$ 

Summing up the above inequalities, we have $0<0$, a contradiction. 

Thus, $M$ is a perfect matching in $G(\bm{q})$. Symmetrically, each perfect matching in $G(\bm{q})$ is a perfect matching in $G(\bm{p})$. Thus $\PM(\bm{p})=\PM(\bm{q})$. 
\end{proof}

\begin{remark}
Lemma \ref{lem:two_prices} is quite interesting in itself. It implies that any two preferred-product graphs, as long as they have at least one perfect matching respectively, must share the same set of all perfect matchings, even though the two graphs can be different. 
\end{remark}


Now we are ready to prove that the space $S=\{\bm{p}\in\mathbb{R}^n \mid \bm{p}\text{ is market-clearing}\}$ is closed under all transformations besides diagonal shifting.

\begin{itemize}
    \item Convex combination. For any $\bm{p}, \bm{q}\in S$, $\alpha\in[0, 1]$, let $\bm{r}=\alpha\bm{p} + (1-\alpha)\bm{q}$. By Lemma \ref{lem:two_prices}, $\bm{p}$ and $\bm{q}$ induce at least one common perfect matching, denoted by $N$. Suppose $i$ is matched to $j$ in $N$, then for any $k=1, ..., n$, we have: 
    \[ v_{ij} - p_j \ge v_{ik} - p_k~\text{ and }~ v_{ij}-q_j \ge v_{ik}-q_k. \]
    Multiply by $\alpha$ and $1-\alpha$, 
    \begin{equation}\label{eqn:convex} \alpha v_{ij} - \alpha p_j \ge \alpha v_{ik} - \alpha p_k~\text{ and }~ (1-\alpha)v_{ij}-(1-\alpha)q_j \ge (1-\alpha)v_{ik}-(1-\alpha)q_k.
    \end{equation}
    Summing the two inequalities in (\ref{eqn:convex}), we get 
    \[ v_{ij} - r_j \ge v_{ik} - r_k. \]
    Thus, $(i, j)$ is in $G(\bm{r})$, for any $(i, j)\in N$. As a result, $N$ is in $G(\bm{r})$, so $\bm{r}\in S$.

    \item Element-wise maximum. For any $\bm{p}, \bm{q}\in S$, let $\bm{r}=(r_1, ..., r_n)$ where $r_i=\max\{p_i, q_i\}$. By Lemma \ref{lem:two_prices}, $\bm{p}$ and $\bm{q}$ induce at least one common perfect matching, denoted by $N$. Suppose $i$ is matched to $j$ in $N$, then for any $k=1, ..., n$, we have: 
    \[ v_{ij} - p_j \ge v_{ik} - p_k~\text{ and }~ v_{ij}-q_j \ge v_{ik}-q_k. \]
    Therefore:
    \[ v_{ij} - r_j = \min\left\{\begin{aligned} v_{ij} - p_j \\ v_{ij} - q_j\end{aligned}\right\} \ge \min\left\{\begin{aligned} v_{ik} - p_k \\ v_{ik} - q_k\end{aligned}\right\} =  v_{ik} - r_k.\]
    Thus, $(i, j)$ is in $G(\bm{r})$, for any $(i, j)\in N$. As a result, $N$ is in $G(\bm{r})$, so $\bm{r}\in S$.
    
    \item Element-wise minimum. This can be proved by switching ``$\max$'' and ``$\min$'' in the above argument for element-wise maximum. Here we present an alternative proof which uses the closure property under diagonal shifting and convex combination. Consider the intersection $\bm{x}$ between the diagonal line passing $\bm{r}$, $\{\bm{r}+t\bm{1}\mid t\in\mathbb{R}\}$, and line $\bm{p}\bm{q}=\{\alpha\bm{p} + (1-\alpha)\bm{q}\mid \alpha\in\mathbb{R}\}$. One can easily verify that $\bm{x}$ lies between $\bm{p}$ and $\bm{q}$, i.e., $\alpha\in[0, 1]$, thus $\bm{x}$ is market-clearing. Then $\bm{r}$ must be market-clearing because it can be obtained by shifting $\bm{x}$ diagonally. 

\end{itemize}

\section{Any Market-Clearing Prices Induce All Maximum Matchings}\label{sec:any_to_all}
In this section we answer the second question: if there are multiple matchings with the maximum weight in a bipartite graph, can each and every maximum matching be induced by some market-clearing price vector?

Mathematically, we ask: whether for any $M\in\opt$, there exists $\bm{p}\in\mathbb{R}^n$, such that $M\in\PM(\bm{p})$? Clearly, this question is in the opposite direction of Property \ref{property:included} ($\PM(\bm{p})\subseteq\opt$) which says that the matching induced by market-clearing prices must be a maximum matching. Our Theorem \ref{thm:all_maximum_matching} gives a positive answer, and even a stronger answer, saying that an \emph{arbitrary} market-clearing price vector satisfies our purpose of finding \emph{all} maximum matchings. 

\begin{theorem}[One versus all]\label{thm:all_maximum_matching}
For any market-clearing price vector $\bm{p}$, 
\begin{equation}
    \emptyset\ne \PM(\bm{p}) = \opt.
\end{equation}
\end{theorem}

In the rest of this section we prove Theorem \ref{thm:all_maximum_matching}.

Note that given Lemma \ref{lem:two_prices} (saying $\PM(\bm{p})=\PM(\bm{q})$ for any market-clearing $\bm{p}, \bm{q}$), all we need to do is to find one specific market-clearing price vector $\bm{p}$ such that $M\in\PM(\bm{p})$ for any given $M\in\opt$. For convenience, we re-label the nodes in $M$ such that $(1, 1), (2, 2), ..., (n, n) \in M$.

We need to solve the the following system of difference constraints: 
\begin{align*}
    \text{find } & (p_1, ..., p_i, ..., p_n) \\
    \text{s.t. } &  v_{11} - p_1 \ge v_{1j} - p_j, ~~ \forall j = 1, ..., n, \\
    & \vdots \\
    & v_{ii} - p_i \ge v_{ij} - p_j, ~~ \forall j = 1, ..., n, \\
    & \vdots \\
    & v_{nn} - p_n \ge v_{nj} - p_j, ~~ \forall j = 1, ..., n.
\end{align*}
Equivalently, 
\begin{align*}
    \text{find } & (p_1, ..., p_i, ..., p_n) \\
    \text{s.t. } &  p_1 \le p_j + v_{11} - v_{1j}, ~~ \forall j = 1, ..., n, \\
    & \vdots \\
    & p_i \le p_j + v_{ii} - v_{ij}, ~~ \forall j = 1, ..., n, \\
    & \vdots \\
    & p_n \le p_j + v_{nn} - v_{nj}, ~~ \forall j = 1, ..., n.
\end{align*}
This is a classic problem and can be solved via shortest paths. Consider a directed graph with $n+1$ nodes $0, 1, ..., n$. The source node $0$ has an directed edge $(0, i)$ with length $0$ to each node $i$. Add an edge from $j$ to $i$ with length $v_{ii} - v_{ij}$, for all $i=1, ..., n$, $j=1, ..., n$. 

Now we compute the length $L(i)$ of the shortest path from $0$ to each node $i$. We allow a path to go through an edge multiple times. So, if there is a cycle with negative total length in this directed graph, then there must be some $L(i)$ that can be $-\infty$ because we can follow that negative cycle infinite times. On the contrary, if there is no such negative cycle, all $L(i)$, $i=1, ..., n$ must be finite, and satisfy all the constraints $L(i)\le L(j) + v_{ii}-v_{ij}$ because there is an edge from $j$ to $i$ with length $v_{ii}-v_{ij}$. Setting  $\bm{p}=(L(1), ..., L(n))$ gives an solution. 

So, it remains to prove that in the directed graph constructed above, there is no negative cycle. Suppose on the contrary, there is an cycle $(j_1\to j_2\to \ldots\to j_{m}\to j_1)$ with negative total length:
\begin{align*}
 (v_{j_2j_2} - v_{j_2j_1}) + (v_{j_3j_3} - v_{j_3j_2}) + ... + (v_{j_mj_m} - v_{j_mj_{m-1}}) + (v_{j_1j_1} - v_{j_1j_{m}}) < 0.
\end{align*}
Rearranging, we get: 
\begin{align*}
 (v_{j_1j_1} + v_{j_2j_2} + ... + v_{j_mj_m}) < (v_{j_2j_1} + ... + v_{j_mj_{m-1}} + v_{j_1j_{m}}).
\end{align*}
Thus, if we replace the sub-matching
$(j_1, j_1), (j_2, j_2), ..., (j_m, j_m)$
in $M$ by the sub-matching
\[(j_2, j_1), ..., (j_m, j_{m-1}), (j_1, j_m),\] then the social welfare of $M$ is strictly increased, contradicting the assumption that $M\in\opt$. This concludes the proof.


\bibliographystyle{plain}
\bibliography{references}

\newpage
\appendix

\end{document}